\documentclass[submission,copyright,creativecommons]{eptcs}
\providecommand{\event}{AFL 2026}
\usepackage[T1]{fontenc}
\usepackage{graphicx}
\usepackage{amsmath}
\usepackage{amssymb}
\usepackage{amsthm}
\usepackage{stmaryrd}

\usepackage{xspace}
\usepackage{forest}

\usepackage{tikz}
\usetikzlibrary{automata, positioning, arrows, arrows}

\usepackage{subcaption}

\usepackage{enumitem}

\usepackage{tikz}
\usepackage{pgfplots}                                          
\pgfplotsset{compat=1.18}   

\usepackage{iftex}

\ifpdf
\usepackage{underscore}         
\usepackage[T1]{fontenc}        
\else
\usepackage{breakurl}           
\fi

\usepackage{ifthen}

\newcommand{\FAT}[1]{\ensuremath{\left(Q^{ #1}, \Sigma, \delta^{ #1}, Q_I^{#1}, Q_F^{ #1}\right)}}
\newcommand{\DFAT}[1]{\ensuremath{\left(Q^{ #1}, \Sigma, \delta^{ #1}, q_I^{#1}, Q_F^{ #1}\right)}}

\newcommand{\TREEALPHABET}{\ensuremath{\Sigma_{\mathcal{T}}}}
\newcommand{\ALL}[2]{\ensuremath{\mathcal{C}^{#1}_{#2}}}
\newcommand{\EX}[1]{\ensuremath{ex #1}}

\newcommand{\AT}{\ensuremath{\mathcal{T}}}
\newcommand{\TL}{\ensuremath{T_L}}
\newcommand{\A}{\ensuremath{\mathcal{A}}}
\newcommand{\B}{\ensuremath{\mathcal{B}}}
\newcommand{\AB}{\ensuremath{\mathcal{AB}}}

\newcommand{\ED}[1]{\ensuremath{\textsc{Edu}_{\leq #1}}\xspace}
\newcommand{\MINDFA}{\ensuremath{\textsc{Min}_{\DFA}}\xspace}
\newcommand{\NATURAL}{\ensuremath{\textsc{Nat}_{\textit{eq}}}\xspace}

\newcommand{\AllCorrect}{\textsc{AllCorrect}\xspace}
\newcommand{\FA}{\textsc{FA}}
\newcommand{\RFA}{\textsc{RFA}}
\newcommand{\NFA}{\textsc{NFA}}
\newcommand{\DFA}{\textsc{DFA}}
\newcommand{\UA}{\textsc{UA}}
\newcommand{\trim}{\textsf{accessible}}
\newcommand{\trimDFA}{\trim\textsc{-DFA}}

\newcommand{\TDNFTA}{\textsc{NtdTA}}
\newcommand{\TDDFTA}{\textsc{DtdTA}}

\newcommand{\LOGSPACE}{\textsc{LOGSPACE}}

\renewcommand{\epsilon}{\varepsilon}
\newcommand{\true}{\textsf{true}}
\newcommand{\false}{\textsf{false}}
\newcommand{\lelex}{\textsf{lelex}}
\newcommand{\inner}{\textsf{inner}}
\newcommand{\Sub}{\textsf{Sub}}
\newcommand{\leaves}{\textsf{leaves}}
\newcommand{\dom}{\textsf{dom}}
\newcommand{\img}{\textsf{img}}

\newcommand{\Dist}{\textsf{Dist}}
\newcommand{\corr}{\rho_\mathcal{AB}}

\newcommand{\corrAlph}{\Sigma^{\mathcal{A}}_{\mathsf{corr}}}

\newcommand{\nerode}[1]{\ifthenelse{\equal{#1}{}}{\sim}{\sim_{#1}}}
\newcommand{\eqclass}[2]{\ensuremath{[#1]_{\ifthenelse{\equal{#2}{}}{}{\nerode{#2}}}}}

\makeatletter
\@tfor\@tempa:=ABCDEFGHIJKLMNOPQRSTUVWXYZ\do{%
	\expandafter\ifx\csname\@tempa\endcsname\relax
	\expandafter\edef\csname\@tempa\endcsname{%
		\noexpand\mathcal{\@tempa}%
	}%
	\fi
}
\makeatother

\newtheorem{proposition}{Proposition}
\newtheorem{theorem}{Theorem}
\newtheorem{definition}{Definition}
\newtheorem{lemma}{Lemma}
\newtheorem{corollary}{Corollary}

\DeclareMathAlphabet{\mathcal}{OMS}{cmsy}{m}{n}

\title{Correcting Deterministic Finite Automata for Didactic Feedback}
\author{Maurice Bornett Herwig 
	\institute{Software Engineering}
	\institute{University of Kassel \\ Kassel, Germany}
	\email{\quad maurice.herwig@uni-kassel.de }
	\and
	Norbert Hundeshagen
	\institute{Theoretical Computer Science / Formal Methods}
		\institute{University of Kassel \\ Kassel, Germany}
	\email{hundeshagen@uni-kassel.de}
}

\newcommand{\titlerunning}{Correcting Deterministic Finite Automata for Didactic Feedback}
\newcommand{\authorrunning}{M. Bornett Herwig \& N. Hundeshagen}

\hypersetup{
	bookmarksnumbered,
	pdftitle    = {\titlerunning},
	pdfauthor   = {\authorrunning},
	pdfsubject  = {EPTCS},               
}

\begin{document}

\maketitle              

    \begin{abstract}
    	Motivated by educational applications, we study the problem of computing all corrections that transform a finite automaton into one recognizing a given regular language $L$. We show that for deterministic finite automata the set of all corrections can be finitely characterized as a regular tree language. The construction is based on a tree encoding of all deterministic finite automata recognizing $L$, which is extended to correction trees that make individual corrections and their induced edit-operations explicit. Leveraging the closure properties of regular tree languages, we introduce so-called filters for selecting corrections satisfying didactic constraints, enabling the derivation of individualized feedback from student submissions.

    \end{abstract}

    \section{Introduction}
    Abstract machines (e.g.\ automata) are fundamental modeling concepts in computer science, with well-known applications in various areas such as circuit design, compiler construction, and natural language processing. Finite automata (\FA) can be seen as the basic principle in this area and therefore are included in standard curricula in theoretical computer science (e.g.\ in \cite{hopcroft2001introduction}).

Typically, the understanding of the syntax and semantics of finite automata as well as the development of basic modeling competencies are supported by exercises requiring students to construct certain automata from a given language description. As the concept of \FA\ is generally considered accessible, exercises of the above type can nevertheless be challenging depending on the given language (cf.\ \cite{pillay2010learning}). 
Educational software designed to support students in overcoming these challenges is widespread. 
The underlying techniques of such tools are manifold. They range from simply verifying correctness via, for instance, an equivalence check against a master solution and providing negative feedback in the form of counterexamples (e.g.\ in \cite{shekhar2014enhancing,bezakova2022effective,robson2024fsm}), to attempts to determine a student's solution intent in order to enable automated personalized feedback (e.g.\ in  \cite{d2015can,d2015automata,d2020automata}). 
We contribute to the latter direction by applying the method of error-correction to \FA\ in order to establish a framework that finitely represents all possible corrections of a given student automaton to a target language. The notion of error-correction is well known in formal language theory, where it describes a method of modifying input data via sequences of edit-operations such as insertions, deletions, and replacements in order to obtain a valid object, for instance a word in a given context-free language (cf.\ \cite{aho1972minimum}) or a tree structure (cf.\ \cite{tai1979treetotree}). From a didactic perspective, such edit-operations can be interpreted as feedback information on the input data (for instance, a solution attempt to an exercise) (cf.\ \cite{bruse2024computing}). As such, they are used in programming education to provide next-step hints to students (e.g.\ in \cite{singh2013automated,rivers2017data}). However, while error-correction traditionally aims at a single correction that is optimal with respect to a given metric (e.g., minimizing the number of edit-operations), this approach seems insufficient in a didactic setting, where different students may follow different solution strategies, that is, different targets a solution attempt may be corrected to. The latter has already been established for context-free languages as the problem of computing all possible ways to correct a given source word into a target context-free language \cite{bruse2024computing,herwig2025earley}. Here we adapt this problem (called \AllCorrect) to finite automata in the following form:
 
\begin{center}
	\begin{tabular}{rl}
		\emph{given}: & a \emph{source} automaton $\A$ and a regular \emph{target} language $L$,\\
		\emph{compute}: & the set $\ALL{\A}{L}$ of all possibilities to \textit{correct} $\A$ such that the language accepted by $\A$ equals $L$.
	\end{tabular}
\end{center}

To this end, we show that for deterministic \FA\ the problem \AllCorrect\ is computable by finitely representing $\ALL{\A}{L}$ as a regular tree language (cf.\ Section~\ref{sec:corrections}). 
From a theoretical perspective, our construction combines classical techniques from automata and formal language theory in an educational setting. 
From this didactic point of view, the finite representation of all corrections of deterministic \FA\ provides a formal framework for applications intended to reason about a student's solution intent in the context of automata exercises (cf.\ Section~\ref{sec:didactics}).

The paper is organized as follows. In Section~\ref{sec:preliminaries} we recall preliminaries on finite automata and finite tree automata. Section~\ref{sec:dfas_and_trees} establishes a tree encoding of deterministic \FA\ and shows that all trees representing automata form a regular tree language. In Section~\ref{sec:corrections}, we prove  the main result of the paper, namely that $\ALL{\A}{L}$ is a regular tree language. In Section~\ref{sec:didactics} we show by example how the regular tree language of all corrections can be used for didactic purposes. Finally, in Section~\ref{sec:conclusion}, we discuss the integration of nondeterministic \FA\ into the presented framework and directions for future work.

    \section{Preliminaries}\label{sec:preliminaries}
    We recall some preliminaries regarding regular languages, finite automata, and finite tree automata.

For that, an \emph{alphabet} $\Sigma$ is a finite, non-empty set of symbols. A \emph{word} $w$ over $\Sigma$ is a finite sequence $w = a_1 \ldots a_{n}$ of symbols. Its length $|w|$ is $n$. The empty word is denoted by $\varepsilon$ and has length $0$. With $\Sigma^*$ we denote the set of all words. A \emph{language} is a subset $L \subseteq \Sigma^*$.

Let $<$ be a total order on $\Sigma$, and let $u = b_1 \cdots b_m$ and $v = a_1 \cdots a_n$. 
We call $u$ a \emph{prefix} of $v$, denoted by $u \prec v$, if there exists a word $x \in \Sigma^*$ such that $v = ux$.
The length-lexicographic order $<_{lelex}$ between words over $\Sigma$ is defined as:
$$
u <_{lelex} v \iff
\begin{aligned}
	& |u| < |v|, \text{ or} \\
	& |u| = |v| \text{ and } \exists i \in \{1, \ldots, min(m, n)\} \text{ such that } b_1 = a_1, \ldots, b_{i-1} = a_{i-1} \text{ and } b_i < a_i.
\end{aligned}
$$

For a function $f$, we denote its domain and image by $\dom(f)$ and $\img(f)$, respectively.

\paragraph*{Finite Automata}
A \emph{nondeterministic finite automaton} (\NFA) $\A = \FAT{}$, where $Q$ is a finite non-empty set of states, $\Sigma$ is an alphabet, $\delta\subseteq Q \times \Sigma \times Q$ is the \emph{transition relation} and $Q_I, Q_F \subseteq Q$ are the sets of \emph{initial states} and \emph{final states}, respectively. 
A \emph{run} of $\A$ on a word $w=a_1\cdots a_n$ is a finite sequence of states $q_0q_1\cdots q_n$ such that $q_0\in Q_I$ and $(q_{i-1}, a_{i},q_{i}) \in \delta$ for all $i \in \{1, \ldots, n\}$. A run is called \emph{accepting} if $q_n \in Q_F$ and the language accepted by $\A$ is ${L(\A) := \{w \in \Sigma^* \mid \text{there is an accepting run of } \A \text{ on } w\}}$.

An \NFA\ is called \emph{accessible} if every state $q \in Q$ is reachable, that is, there is a run from some initial state to $q$. 
A state $p \in Q$ is called a \emph{sink state} if $p \notin Q_F$ and satisfies $\delta(p,a)=p$ for all $a\in\Sigma$.

An \NFA\ is a \emph{deterministic finite automaton} (\DFA) if $\delta$ is right-unique and left-total, that is, a function $\delta\colon Q\times\Sigma\to Q$ and $|Q_I| = 1$. Since $\delta$ is a function, we extend it to runs and write $\hat{\delta}(p, w)=q$ to express that $q$ is reached in a run from $p$ reading $w$. In the following, when dealing with \DFA\ we treat $\delta$ as either a function or a relation, depending on the context.

Let $L \subseteq \Sigma^*$ be a language. The \emph{Myhill-Nerode relation} $\nerode{L}$ is defined by 
$$w' \nerode{L} w \iff \forall u \in \Sigma^*: w'u \in L \iff wu \in L.$$ It partitions $\Sigma^*$ into equivalence classes $\eqclass{w}{L} = \{u \in \Sigma^* \mid u \nerode{L} w\}$. With $\setminus_{\nerode{L}}$ we denote the quotient of a structure under $\nerode{L}$, e.g.  $\Sigma^* \setminus_{\nerode{L}} = \{\eqclass{w}{L} \mid w \in \Sigma^* \}$.
To increase readability we simply write $\nerode{}$ and $\eqclass{u}{}$ for the Nerode-relation and its equivalence classes, when $L$ is clear from the context.
If $L$ is regular, then $\nerode{}$ has finite index. Based on $\nerode{}$ one can compute a unique \emph{minimal} \trimDFA\ $\A_{\min}$ for every regular language $L$, where each state corresponds exactly to one equivalence class of $\nerode{}$.
Furthermore for any \trimDFA\ accepting $L$, there is a surjective homomorphism onto $\A_{\min}$
(cf.~\cite[Section~4.4]{hopcroft2001introduction}).

\paragraph*{Finite Tree Automata}
We here only present a brief summary of definitions and refer to \cite{TATA} for a comprehensive introduction. A \emph{ranked alphabet} is a pair $\left(\Sigma, \sigma_{\Sigma}\right)$, where $\Sigma$ is, as in the case of words, a finite, non-empty set of symbols, and $\sigma_{\Sigma}: \Sigma \to \mathbb{N}$ is a function that assigns to each symbol a natural number, called its \emph{rank}. 
A (finite) \emph{tree-domain} is a prefix-closed subset $D \subseteq \mathbb{N}^*$ such that for all $v \in \mathbb{N}^*$ and $i \in \mathbb{N}$: 
$$v(i+1) \in D \Rightarrow vi \in D.$$ Observe that $vi$ and $v(i+1)$ refer to the concatenation of natural numbers. 
A \emph{$\Sigma$-tree} is a function ${t\colon D\to \left(\Sigma, \sigma_{\Sigma}\right)}$ (or simply $t\colon D\to \Sigma$) that assigns to each node $v \in D$ a symbol (or label) $t(v) \in \Sigma$ such that, for all $v \in D$ with $\sigma_{\Sigma}(t(v)) = k$, it holds that $v(k-1) \in D$ and $v(k) \notin D$.

Furthermore with $\operatorname{\inner}(t) := \{v \in D \mid \sigma_{\Sigma}(t(v)) > 0\}$ we denote the set of all non-leaf nodes of $t$ and with $\operatorname{\leaves}(t) := D\setminus\operatorname{\inner}(t)$ the leaves of $t$.

A \emph{path} in a $\Sigma$-tree $t$ is a sequence $t(v_1)\cdots t(v_n) \in \Sigma^*$ such that $v_1,\dots,v_n \in D$ satisfying $v_1 = 0$ and, for all $1 \leq i < n$, $v_{i+1} = v_i  j$ for some $j < \sigma_{\Sigma}(t(v_i))$.

A \emph{top-down nondeterministic finite tree automaton} (\TDNFTA) is a tuple $\AT = \left(Q, F, Q_I, \delta \right)$, where $Q$ is a finite, non-empty set of states, $F = \left(\Sigma, \sigma_{\Sigma}\right)$ is a ranked alphabet, $Q_I \subseteq Q$ is the set of \emph{initial states}, and $\delta = \bigcup_{a \in \Sigma} \delta_a$ is the \emph{transition relation} with $\delta_a: Q \to 2^{Q^{\sigma_{\Sigma}(a)}}$. If $\sigma_{\Sigma}(a) = 0$, we simply write $Q \to \{\true, \  \false\}$ instead of $Q \to \{\{\}, \{()\}\}$.  

A \emph{run} of $\AT$ on a $\Sigma$-tree $t\colon D \to \Sigma$ is a function $p: D \to Q$ such that $p(\varepsilon) = q_I$ for $q_I \in Q_I$, and for all $v \in D$ with $t(v) = a \in \Sigma$ and $\sigma_{\Sigma}(a) = k$, it holds that: 
$$\left(p(v0), \ldots, p(v(k-1)) \right) \in \delta_a (p(v)).$$
A run is \emph{accepting}, if $\delta_a (p(v)) = true$ for each leaf $v$ of the tree. 
The tree-language accepted by $\AT$ is defined as: 
$L(\AT) = \{t \in T_{\Sigma} \mid $ there exists an accepting run of $\AT \text{ on } t\}$.
A tree-language is called \emph{regular} if it is accepted by some \TDNFTA. 

An \TDNFTA is called a \emph{deterministic finite tree automaton} (\TDDFTA), if $|\delta_a(q)| \leq 1$ for all $q \in Q, a \in \Sigma$. 
In contrast to finite automata, \TDDFTA\ are less expressive than \TDNFTA.
    
    \section{DFAs and Trees}\label{sec:dfas_and_trees}
    Again, for a given regular language $L$ we are interested in correcting a given finite automaton $\A$ to all target automata $\B$ with $L(\B)=L$. Thus, according to the formulation of the \AllCorrect problem in the introduction, we will establish a finite representation of all finite automata accepting $L$. This representation is introduced for deterministic finite automata, whereas later we discuss possible extensions onto nondeterministic versions (see Section~\ref{sec:conclusion}).

In the following we fix an alphabet $\Sigma = \{x_1,\dots, x_k\}$, a total order $x_1 < \dots < x_k$ of alphabet symbols, and a regular language $L\subseteq \Sigma^*$. Furthermore, we heavily use properties of the  Nerode-equivalence $\nerode{L}$ and overload the equivalence class notation $\eqclass{v}{L}=\{w\mid v\nerode{L} w\}$, by either referring to a representative $u$ of the class, here the length-lexicographic shortest word $u<_\lelex v$, or to the equivalence class itself as a set.

We start by stating a series of implicitly given definitions that are intended to follow the standard method of encoding a \DFA\ as a tree in form of a prefix tree unfolding of its transition graph. For that, let $(\TREEALPHABET,\sigma_{\TREEALPHABET})$ be a ranked tree alphabet with $\TREEALPHABET=\{\eqclass{u}{L}, \bot_{\eqclass{u}{L}} \mid u \in \Sigma^*\}$ and 
$\sigma_{\TREEALPHABET}(\eqclass{u}{L}) = |\Sigma|$ and $\sigma_{\TREEALPHABET}(\bot_{\eqclass{u}{L}}) = 0$ for all $u \in \Sigma^*$. Observe that since $L$ is regular $\TREEALPHABET$ is finite. 

Let $\A = \DFAT{\A}$ be an \trimDFA\ with $L(\A)=L$ and let $f_{\lelex}\colon Q^{\A}\to\Sigma^*$ be the function 
$$f_{\lelex}(q) := min_{<_\lelex}\left(\left\{w \in \Sigma^* \mid  \hat{\delta^{\A}}(q_I^{\A}, w) = q\right\}\right)$$
 that associates every state $q$ of $\A$ with the length-lexicographic shortest word reaching that state. Note, that
  $f_{\lelex}\left(Q^{\A}\right)=\left\{f_{\lelex}(q)\mid q\in Q^{\A}\right\}$
  is prefix-closed and thus, a tree domain. We further extend $f_{\lelex}(Q^{\A})$ to the domain $D:= f_{\lelex}(Q^{\A}) \cup \{wx\mid x\in\Sigma, w\in f_{\lelex}(Q^{\A})\}$ to encode the leaves of the tree unfolding of $\A$.
  
  For $\A$ we define a $\TREEALPHABET$-tree $t_\A\colon D \to \TREEALPHABET$ as follows: 
  \begin{align*}
    	t(w) := \begin{cases}
  		\eqclass{w}{L}, & \textit{ if }w\textit{ is not a leaf in } D \textit{ or,} \\
  		\bot_{\eqclass{w}{L}},  & \textit{ else.}
  	\end{cases}
  \end{align*} Figure~\ref{fig:example_state_tree} shows an example. 
  Note, that this definition essentially uses the fact, that there is a surjective homomorphism from $\A$ onto the Nerode-(resp.~minimal-)automaton for $L$ (cf.~Section~\ref{sec:preliminaries}).
  Thus, $t_\A$ can basically be interpreted as a tree of all runs reaching any state in $\A$ once, mapped onto the minimal \DFA.   
  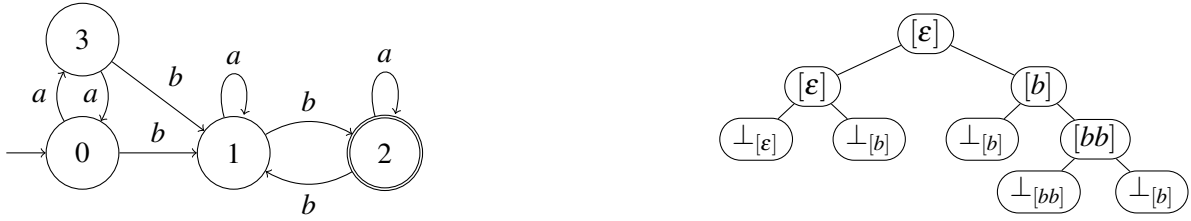
\begin{figure}[]
  	\centering
  	\begin{minipage}{0.45\textwidth}
  		\begin{tikzpicture}[auto]
  			\node[state, initial, initial text={}] (q0) at (0,0) {$0$};
  			\node[state] (q1) at (2,0) {$1$};
  			\node[state, accepting] (q2) at (4,0) {$2$};
  			\node[state] (q3) at (0, 1.5) {$3$};
  			
  			\path[->]
  			(q0) edge[bend left, left] node{$a$} (q3)
  			(q3) edge[bend left, left] node{$a$} (q0)
  			(q3) edge[] node{$b$} (q1)
  			(q0) edge[above] node{$b$} (q1)
  			(q1) edge[loop above] node{$a$} (q1)
  			(q1) edge[above, bend left] node{$b$} (q2)
  			(q2) edge[below, bend left] node{$b$} (q1)
  			(q2) edge[loop above] node{$a$} (q2);
  			
  		\end{tikzpicture}
  	\end{minipage}
  	\hfill
  	\begin{minipage}{0.4\textwidth}
  		\centering
  		\resizebox{\linewidth}{!}{
  			\begin{tikzpicture}[
  				level distance=7mm,
  				level 1/.style={sibling distance=30mm},
  				level 2/.style={sibling distance=15mm},
  				level 3/.style={sibling distance=15mm},
  				every node/.style={
  					draw,
  					rounded rectangle,
  					align=center,
  					inner sep=2pt
  				},
  				empty/.style={
  					draw=none,
  					inner sep=0pt,
  					minimum size=1pt
  				},
  				dotted child/.style={
  					edge from parent/.style={draw,dotted}
  				},
  				edge from parent/.style={draw}
  				]
  				
  				\node {$\eqclass{\varepsilon}{}$}
  				child {
  					node  {$\eqclass{\varepsilon}{}$}
  					child { node {$\bot_{\eqclass{\varepsilon}{}}$}}
  					child { node {$\bot_{\eqclass{b}{}}$}}
  				}
  				child {
  					node  {$\eqclass{b}{}$}
  					child { node {$\bot_{\eqclass{b}{}}$} }
  					child {
  						node {$\eqclass{bb}{}$}
  						child { node{$\bot_{\eqclass{bb}{}}$} }
  						child { node  {$\bot_{\eqclass{b}{}}$} }
  					}
  				};
 
  			\end{tikzpicture}
  		}
  	\end{minipage}
  	\caption{An exemplary \DFA\ $\A_{\EX{1}}$ (left) for the language $L_{\A_{\EX{1}}}=\{w\in\{a,b\}^*\mid |w|_b \geq 2 \textit{ and even}\}$ and its representation as a $\TREEALPHABET$-tree $t_{\A_{\EX{1}}}$ (right) with the tree alphabet $\TREEALPHABET$ comprised of the Nerode-classes of $L_{\A_{\EX{1}}}$ given as $[\varepsilon]$, $[b]$, and $[bb]$.}
  	\label{fig:example_state_tree}
  \end{figure}
 
 The following lemma formalizes the connection between $\A$ and the minimal \DFA\ for $L$ by essentially showing that the tree induced by the minimal \DFA\ of $L$ is a prefix of the tree representation of every \DFA\ accepting the same language. 

\begin{lemma}\label{lemma:prefix-tree}
	 	Let $L$ be a regular language and  $\A_{\min} = \DFAT{\min}$ be the minimal \DFA\ for $L$. 
	 	Then for every \DFA\ $\A = \DFAT{\A}$ with $L(\A) = L$ it holds that 
	 	$$f_{\lelex}\left(Q^{\min}\right) \subseteq f_{\lelex}\left(Q^{\A}\right).$$
\end{lemma}

\begin{proof}
	Recall that there is a surjective homomorphism $h: Q^{\A} \to Q^{\min}$. 
	Let $q \in Q^{\min}$ and ${u = f_{\lelex}(q)}$, then there is a $p \in Q^{\A}$ with $h(p)=q$ and $\hat{\delta}^{\A}(q_I^{\A}, u)=p$. It remains to show that $f_{\lelex}(p) = u$, that is,
	$u$ is also the length-lexicographically shortest word reaching $p$ in $\A$.
	
	Suppose the contrary that there is $v <_{\lelex} u$ with $\hat{\delta}^{\A}(q_I^{\A}, v)=p$. Since $h$ is a homomorphism with $h(p)=q$ this would imply that $\hat{\delta}^{\min}(q_I^{\min}, v)=q$, contradicting the minimality of $u=f_{\lelex}(q)$.
\end{proof}
Finally, we define the set of tree representations of all \trimDFA\ for a given regular language $L$ as $$\TL:=\{t_\A\mid \A\textit{ is an \trimDFA\ with } L(\A)=L\}.$$ 

\begin{proposition}\label{prop:dfa_reg_tree_lang}
	Let $L\subseteq \Sigma^*$ be a regular language, then $\TL$ is a regular tree language.
\end{proposition}

\begin{proof}
	The proof follows from the definition of trees and the fact that for each \trimDFA\ for $L$ there is a surjective homomorphism onto the unique minimal \DFA\ $\A_{\min}$.
	Therefore, the main idea is that the tree automaton $\mathcal{T}$ for $\TL$ simulates a run of
	$\A_{\min}$ along every root-to-leaf path of a tree, 
	while ensuring that the prefix tree induced by the minimal \DFA\ for $L$ is embedded in the tree, as required by Lemma~\ref{lemma:prefix-tree}. 

	For that, let $\Sigma=\{x_1,\dots, x_k\}$ and $\A_{\min} = \DFAT{\min}$ be the minimal \DFA\ for $L$.
	Furthermore, let $Q^{\min}_{\leaves} = \left\{ q \in Q^{\min} \mid \nexists p\in  Q^{\min} \textit{ such that } f_{\lelex}(q) \prec f_{\lelex}(p) \right\}$ denote the subset of states corresponding to the leaves of the tree domain induced by $f_{\lelex}(Q^{\min})$.
	From $\A_{\min}$ we construct a \TDDFTA\ 
	$$\AT = \left(Q^{\min}\times\left(\Sigma \cup\{\varepsilon\}\right)\times \{\true, \false\}, \left(\TREEALPHABET,\sigma_{\TREEALPHABET}\right), I, \delta\right)$$ 
	as follows.
	Let $x\in\Sigma\cup\{\varepsilon\}$, $u\in\Sigma^*$ then	
	\begin{align*}
		I := \begin{cases}
			 \left\{\left(q_I^{\min},\varepsilon,\true \right)\right\}, & \textit{ if } |Q^{\min}| = 1, \\ 
			 \left\{\left(q_I^{\min},\varepsilon,\false\right)\right\}, & \textit{ else.}
		\end{cases}
	\end{align*}
	Furthermore, let $q_i = \delta^{\min}(p,x_i)$ for  $i \in \{1,\dots,k\}$, then
	
	$$
		\delta_{\eqclass{ux}{}}\big((p,x,b)\big)
		:= 		
		\big(
		(q_1,x_1,b'_1), \dots, (q_k,x_k,b'_k)
		\big),		
	$$
	where $b'_i = \true$, if $b = \true$ or $q_i \in Q^{\min}_{\leaves}$, and $b'_i = \false$ otherwise.
	
	And finally, 	
	$$\delta_{\bot_{\eqclass{ux}{}}}\big((p,x,b)\big):=  \true \text{ if } b=\true.$$
	Thus, the second component $x$ of a state $(p,x,b)$ ensures that the transition $\A_{\min}$ has taken to reach $p$ corresponds to the equivalence class the current node is labeled with.	
	Furthermore Lemma~\ref{lemma:prefix-tree} states that each tree must contain a prefix corresponding to the minimal \DFA. This condition is tracked by the component $b$, which records whether the corresponding region of the prefix tree has already been completely traversed along the current path. In particular, $b$ is switched from $\false$ to $\true$ upon reaching a leaf of this prefix tree, i.e., an element of $Q^{\min}_{\leaves}$.	

	It remains to show, that $\TL=L\left(\AT\right)$. The direction from left to right follows immediately from the definition of the tree $t_\A$ for any \trimDFA\  $\A$ and the fact that there is a surjective homomorphism from $\A$ onto $\A_{\min}$. As $t_{\A}\in \TL$ every node of $t_\A$ is labeled with the Nerode class $\eqclass{w}{}$ corresponding to a length-lexicographic shortest word $w$ reaching a state of $\A$, and for each such $\eqclass{w}{}$ the $x_i$-successor corresponds to extending the word $w$ by $x_i$. Hence, along every root-to-node path $u$ the first component of $\AT$ simulates $\mathcal \A_{\min}$ on $u$.
	Moreover, since $\A$ is \trim, every Nerode-class occurs at least once among the inner nodes, so the component $b$ can mark all states from $Q^{\min}_{\leaves}$, and therefore $\AT$ accepts $t_{\A}$, i.e. $t_{\A}\in L\left(\AT\right)$.
	
	In the other direction let $t\in L\left(\AT\right)$ be a tree accepted by $\AT$. From $t$ we construct an \trimDFA\ $\A^t = \DFAT{t}$ as follows.
	Recall that $\operatorname{\inner}(t)$ is the set of inner nodes of $t$,  
	then $Q^t:=\inner(t)$,
	 $q^t_I:=\varepsilon$, and $Q^t_F:=\{w\mid t(w) \in L\}$.
	 For the latter recall  that the node label $t(w)$ is a representative of the equivalence class $\eqclass{w}{}$ over $L$.
	
	By definition of $Q^t$ every inner tree node, which is actually a word $w\in \Sigma^*$, is treated as a state, thus, $\delta^t$ needs to be compatible with the respective tree structure until reaching a state $u$ indicating that its $x$-child is a leaf in $t$ labeled with $\bot_{\eqclass{ux}{}}$. Recall that $\bot_{\eqclass{ux}{}}$ implies a run of a \DFA\ ending in a state $\eqclass{ux}{}$ already present in $t$. 
	Hence, $\A^t$ changes its state to one of the equally labeled states. Formally, let $u\in Q^t$, $x\in\Sigma$, $\operatorname{\dom}_v(t):=\{w\mid t(w)=v\}$ the projection of $\dom(t)$ onto a specific label $v$, and $w\in\operatorname{\dom}_{\eqclass{ua}{}}(t)$ one node labeled $\eqclass{ua}{}$, then 
	$$\delta^t(u,a):=\begin{cases}
		ua, \text{ if } ua\in\operatorname{\inner}(t);\\
		w, \text{ else}.
	\end{cases}$$
	It remains to show, that $L(\A^t)=L$. For that, let $h\colon Q^t\to Q^{\min}$ be defined as $h(u):=t(u)$, that is, each state of $\A^t$ (resp. node of $t$) is mapped onto its equivalence class, i.e.~ the state of the minimal automaton $\A_{\min}$ for $L$. Observe that $h$ is a surjective homomorphism with the condition $p\in Q^t_F \iff h(p)\in Q^{\min}_F$. Hence, $L(\A^t)=L$ holds.
\end{proof}
It is worth noting that the above characterization of $\TL$ as a regular tree language enables a reduction from the equivalence problem of \DFA\ onto the word problem of \TDDFTA, which is \LOGSPACE-complete (e.g.~in \cite{TATA}). However, as the reduction requires computing the minimal automaton, it lies well within known complexity bounds.

    \section{Corrections on Finite Automata}\label{sec:corrections}
    Recall that corrections on words are defined as a sequence of edit-operations acting on the linear structure of a source word (e.g., in \cite{aho1972minimum}). Corrections on non-linearly structured data, such as trees or graphs, can be defined similarly as a sequence of edit operations on nodes and edges. 
As mentioned in the introduction, we are here interested in a general approach, that is, we start by establishing the notion of a correction in the form of a mapping between a source automaton and a target \trimDFA. Note that the latter essentially reverses the ideas presented in \cite{tai1979treetotree}, where the author has shown that a sequence of edit operations transforming one tree into another defines a mapping between the two. 
We then extend this idea by showing that all corrections of a source automaton to possible targets can be represented finitely as a regular tree language. As discussed later (cf.\ Section~\ref{sec:didactics}), this can be used to establish a notion of students' intent by collecting feedback information in the form of edit-operations and filters.

We start with a very general notion of correction by formalizing structural differences between two automata via a function on their state sets, which then essentially serves as a partial embedding of the source automaton into the target.
\begin{definition}\label{def:corrections}
Let $\A = \FAT{\A}$ and $\B = \FAT{\B}$ be \NFA\ with $Q^\A\cap Q^\B=\emptyset$. A (possibly partial) injective mapping $\corr\colon Q^\A\to Q^\B$ 
  is called a \textit{correction} and the automaton corrected through $\corr$,  $\AB=\FAT{\AB}$ is defined as follows:		
		\begin{align*}
			Q^{\AB}  = &\;\dom(\corr) \cup \left(Q^\B\setminus\img(\corr)\right)\\
			\delta^{\AB} = &\; 
			\left\{(p,a,p')\mid p,p'\in\dom(\corr),\ \left(\corr(p),a,\corr(p')\right)\in\delta^\B\right\} \\
			& \cup \left\{(p,a,q)\mid p\in\dom(\corr), q\in\left(Q^\B\setminus\img(\corr)\right), (\corr(p),a,q)\in\delta^\B\right\}\\
			& \cup \left\{(q,a,p)\mid q\in\left(Q^\B\setminus\img(\corr)\right), p\in\dom(\corr), (q,a,\corr(p))\in\delta^\B\right\}\\
			& \cup \left\{(q,a,q')\mid q,q'\in(Q^\B\setminus\img(\corr)),\ (q,a,q')\in\delta^\B\right\}\\
			Q_{I}^{\AB}  = &\; \left\{p\in Q^{\A}\mid \corr(p)\in Q_{I}^{\B}\right\} \cup \left\{q\in (Q^\B\setminus\img(\corr)) \mid q\in Q_{I}^{\B}\right\} \\
			Q_{F}^{\AB}  = &\; \left\{p\in Q^{\A}\mid \corr(p)\in Q_{F}^{\B}\right\} \cup \left\{q\in (Q^\B\setminus\img(\corr)) \mid q\in Q_{F}^{\B}\right\} \\
		\end{align*} 
		Moreover, we also write $\corr(\A)$ to denote the corrected automaton $\AB$. 
\end{definition}
Note that the latter definition essentially states that the corrected automaton $\AB$ of a source automaton $\A$ is, up to renaming, the target automaton itself. Hence, the following corollary is an immediate consequence of Definition~\ref{def:corrections}.

\begin{corollary}\label{cor:isomorph}
	Let $\A$, $\B$ be finite automata, $\corr$ a correction, and $\AB$ the corrected automaton, then $\AB$ is isomorphic to $\B$.
\end{corollary}

Nevertheless, Definition~\ref{def:corrections} already implies certain types of edit-operations (e.g.~insertion and deletion of states and transitions) associated with a correction $\corr$. 
In the following the above definition of correction will form the basis for introducing a framework to finitely represent and compute all corrections of a given source automaton into all possible deterministic target automata defining a given target language. 
  
We start by stating a further consequence of Corollary~\ref{cor:isomorph}, that is, if the target automaton $\B$ is an \trimDFA, so is $\corr(\A)$ for any finite automaton $\A$. Hence, from here on, we rely on the notation established in Section~\ref{sec:dfas_and_trees}. We fix an alphabet $\Sigma$, a regular language $L \subseteq \Sigma^*$, a source automaton to be corrected, not necessarily deterministic, $\A = \FAT{\A}$, and a target \trimDFA\ ${\B = \FAT{\B}}$.  With $\B$ we associate its tree unfolding $t_\B\colon D_\B\to \TREEALPHABET$, where $\left(\TREEALPHABET,\sigma_{\TREEALPHABET}\right)$ denotes the corresponding tree alphabet and $D_\B$ the tree domain. Recall from Section~\ref{sec:dfas_and_trees} that $D_\B$ is defined via the mapping $f_{\lelex}$, which assigns to each state of $\B$ the length-lexicographically shortest word reaching that state. Finally, we write $\TL$ for the set of all tree unfoldings of \trimDFA{}s accepting $L$.

We start by introducing a tree representation of a correction $\corr$ by annotating every node in the target automatons tree representation $t_\B$ with the corresponding states of the source $\A$.

\begin{definition}\label{def:corr_tree}
	Let $\A$ be the source automaton, $\B$ the \trimDFA\ target, $t_\B\colon D_\B \to \TREEALPHABET$ its tree unfolding, and $\corr$ a correction as fixed above. Furthermore, let $\corrAlph:=\TREEALPHABET\times\left(Q^{\A}\cup\{\#\}\right)$ with
	 $\sigma_{\corrAlph}((v,m)):=\sigma_{\TREEALPHABET}(v)$ be a new tree alphabet. A \textit{correction tree of $\B$ under $\corr$}, $t^{\corr}_{\B}\colon  D_\B \to \corrAlph$, is defined as a relabeling of $t_\B$, that is for every $x\in D_\B$, 
	\begin{align*}
		t^{\corr}_{\B}(x) = \begin{cases}
			(t_{\B}(x), q), &\text{if there exists } q\in\dom(\corr) \text{ such that } f_\lelex(\corr(q))=x,\\
			(t_{\B}(x), \#), &\text{else.}
		\end{cases}
	\end{align*}
\end{definition}
An example of a correction tree is shown in Figure~\ref{fix:example_correction_tree}. 
\begin{figure}
		\begin{subfigure}[b]{0.3\textwidth}
			\centering
			\begin{tikzpicture}[node distance=16mm, every state/.style={minimum size=4mm, inner sep=2pt}]
		\node[state,initial,initial text= ] (q0) at (-1,0) {$q_0$};
		\node[state,accepting] (q1) at (0.5,0) {$q_1$};
		\node[state] (q2) at (0,-1) {$q_2$};
		
		\node[state,initial,initial text= ] (p0) at (-4,-1) {$p_0$};
		\node[state,accepting] (p1) at (-3,0) {$p_1$};
		\node[state,accepting] (p2) at (-4.5,0) {$p_2$};
		
		\node at (-4, -1.75) {$\A$};
		\node at (0, -1.75) {$\B_{\EX{2}}$};
		
		\path[->] 
		(q0) edge [above] node {$0$} (q1)
		(q0) edge [below] node {$1$} (q2)
		(q1) edge [bend left, right] node {$a$} (q2)
		(q2) edge [left] node {$a$} (q1)
		(p0) edge [above] node {$a$} (p1)
		(p1) edge [bend left, right] node {$1$} (p0)
		(p1) edge [above] node {$a$} (p2)
		(p2) edge [loop left] node {$a$} (p2)
		
		(p1) edge[dashed, blue, bend right] node[below] {$\corr(p_1)$} (q0)
		(p0) edge[dashed, blue, bend right] node[below] {$\corr(p_0)$} (q2)
		;
	\end{tikzpicture}
\end{subfigure}
\hfill
\begin{subfigure}[b]{0.55\textwidth}
	\centering
	\resizebox{\linewidth}{!}{
		\begin{tikzpicture}[
			level distance=8mm,
			level 1/.style={sibling distance=42mm},
			level 2/.style={sibling distance=14mm},
			every node/.style={
				draw,
				rounded rectangle,
				align=center,
				inner sep=2pt
			},
			empty/.style={
				draw=none,
				inner sep=0pt,
				minimum size=1pt
			},
			dotted child/.style={
				edge from parent/.style={draw,dotted}
			},
			edge from parent/.style={draw}
			]
			
			\node (e) {$\eqclass{\varepsilon}{}, p_1$}
			child {
				node (p1) {$\eqclass{0}{}, \#$}
				child[dotted child] { node[dotted] {$\bot_{\eqclass{a}{}},\#$} }
				child[dotted child] { node[dotted] {$\bot_{\eqclass{a}{}},\#$}  }
				child { node (b1) {$\bot_{\eqclass{1}{}},\#$} }
			}
			child {
				node (p2) {$\eqclass{1}{}, p_0$}
				child[dotted child] { node[dotted] {$\bot_{\eqclass{a}{}},\#$} }
				child[dotted child] { node[dotted] {$\bot_{\eqclass{a}{}},\#$} }
				child { node (b1) {$\bot_{\eqclass{0}{}},\#$} }
			}
			child[dotted child] { node[dotted] (a) at(-2.5, 0) {$\bot_{\eqclass{a}{}},\#$}  };
		\end{tikzpicture}
}
\end{subfigure}
	\caption{Two exemplary \DFA{s}, $\A$ and $\B_{\textsf{ex2}}$, are shown on the left, together with a correction $\corr$ defined by $\corr(p_0) = q_2$ and $\corr(p_1) = q_0$, illustrated by the blue dotted lines. For clarity, we omit the sink state in both \DFA{s}. Furthermore, $\B_{\textsf{ex2}}$ is the minimal \DFA\ for the language ${L_{\A_\textsf{ex2}} = \{xa^n \in  \{0, 1, a\}^{*}\mid x=0, \text{ if } n \text{ is even and } x = 1 \text{ otherwise}\}}$. The correction tree $t^{\corr}_{\B{\textsf{ex2}}}$ for $\A$, $\B_{\textsf{ex2}}$, and $\corr$ (with the chosen alphabet ordering $0 < 1 < a$) is shown on the right. Nodes and edges associated with the sink state are dotted.}
	\label{fix:example_correction_tree}
\end{figure}
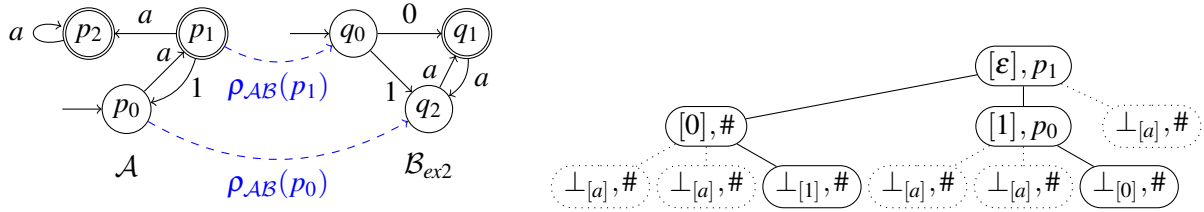
Hence, the latter definition simply adds the state $q\in\dom(\corr)$ to the inner node of $t_\B$ representing $\corr(q)$.
However, in contrast to Definition~\ref{def:corrections}, the tree representation $t^{\corr}_{\B}$ of $\corr$ loses information in the sense that several corrected automata $\mathcal{AB}$ may be obtained from it, depending on how the leaves $\bot_v$ of $t^{\corr}_{\B}$ are interpreted (cf.~ Proposition~\ref{prop:dfa_reg_tree_lang}).

The latter definition and in particular the observation that for a given tree unfolding $t_\B$ there are only finitely many correction trees $t^{\corr}_{\B}$ representing all possible domains of corrections $\corr$, gives rise to the main result of this section.

\begin{theorem}\label{thm:all_correct}
	Let $\A$ be a finite automaton and $L$ a regular language, then $$\ALL{\A}{L}:=\{t^{\corr}_{\B}\mid \B\textit{ is an \trimDFA\ with } L(\B)=L \text{ and } \corr \text{ is a correction}\}$$ is a regular tree language.
\end{theorem}

\begin{proof}
	(\textit{sketch}) The proof is a straightforward extension of the one given for Proposition~\ref{prop:dfa_reg_tree_lang}. To avoid repeating technical details, we here only sketch the changes to the tree automaton $\AT$ constructed for $\TL$.
	
	Recall that $\AT$ makes use of the fact that for every \trimDFA\ for $L$ there is a surjective homomorphism onto the minimal \DFA\ $\A_{\min}$ for $L$. It does so by ensuring that runs along every path of $t$ are compatible with $\A_{\min}$ and that the number of inner nodes of $t$ is at least as large as the number of Nerode classes.	We now extend this automaton to an \TDNFTA\ $\AT^{\corr}$
	over the alphabet $\corrAlph:=\TREEALPHABET\times\left(Q^{\A}\cup\{\#\}\right)$ (see Definition~\ref{def:corr_tree}).

	The automaton $\AT^{\corr}$ simulates $\AT$ on the first component of each label.
	It remains only to check that the second component indeed describes a partial injective function from the source automaton $\A$ to the represented target automaton. More precisely, for every $q\in Q^{\A}$ and every $v\in\TREEALPHABET$, there is at most one inner node labeled with $(v,q)$ in every accepted tree.
	
	Since $Q^{\A}$ is finite we simply store subsets of $Q^{\A}$ in the states of $\AT^{\corr}$. 
	The automaton then guesses at each node the set of states from $Q^{\A}$ that occur in the subtree rooted at that node. 
	It then verifies locally that the sets guessed for the children are correct and that, if the current node is labeled by some $q\in Q^{\A}$, then $q$ does not occur in any child subtree. For leaves, the automaton additionally checks that the second component is $\#$. 
	
	Thus, $\AT^{\corr}$ accepts all trees of $\TL$ labeled with states from $Q^{\A}$ encoding a partial injective function $\corr$. Hence, $\ALL{\A}{L}$ is a regular tree language.	
\end{proof} 

Theorem~\ref{thm:all_correct} is a solution for the \AllCorrect-problem stated in the introduction. However, note that strictly, if you define a correction as a sequence of edit-operations, all corrections for instance may also contain multiple sequences of deleting and adding same components in the automaton. In that sense, the solution presented here can be seen as normalized sequences of edit-operations restricted by the size of the target automaton.

    \section{A Didactical Perspective on Corrections: Edit-Operations  and Filters}\label{sec:didactics}    
    Section~\ref{sec:corrections} lays the theoretical foundations of corrections on finite automata. There, Theorem~\ref{thm:all_correct} offers a constructive solution to the \AllCorrect problem by computing a tree automaton $\AT^{\corr}$ that accepts the language of all possible corrections of an \NFA\ into a regular language $L$, given by its minimal \DFA. Our main focus here is how to leverage the seemingly arbitrary information stored in $\AT^{\corr}$. Therefore, two questions arise from that and the motivation described in the introduction, namely, \textit{what kind of feedback information can be computed from a given correction tree}, and \textit{how can we filter the finite representation of all correction trees $\AT^{\corr}$ for certain candidates}. 

\paragraph{Edit-Operations}

Recall that our goal is to present certain feedback information to student's solution on how to improve their attempt. We encode this information by so-called edit-operations (e.g.~inserts and deletions of states and transitions). Note that the latter is already implied in Definition~\ref{def:corrections}. We begin with explicitly defining these edit-operations for a given correction tree.

\begin{definition}\label{def:edit_ops}
	Let $\A=\FAT{A}$ be the source automaton, $L$ a regular language, and $t\in T^{\mathcal{A}}_L$ a correction tree in the tree language of all corrections of $\A$ over the alphabet $\corrAlph$.
	With $t$ we associate the following edit-operations on states of $\A$: 
			\begin{itemize}
					\item every state $q\notin \{p\in Q^\A\mid \exists x\in \inner(t)\text{ with } t(x)=(v,p)\}$ defines a \textbf{deletion of state} $q$ from $\A$, 
					\item every state $q\in \{p\in Q^\A\mid \exists x\in \inner(t)\text{ with } t(x)=(v,p)\}$ defines a \textbf{kept state} $q$ of $\A$, and  
					\item every node $x\in\inner(t)$ with $t(x)=(v,\#)$ defines an \textbf{insertion of state} $x$ into $\A$. 
					
			\end{itemize}
	As for the transitions of $\A$ we associate the following edit-operations:
			\begin{itemize}
				\item every transition $(p,a,q)\in\delta^{\A}$ for which
				\begin{enumerate}[label=(\arabic*)]
					\item there is no consecutive pair of nodes $x,xa\in\inner(t)$ with $t(x)=(v,p)$ and $t(xa)=(w,q)$ defines a \textbf{deletion of transition} $(p,a,q)$ from $\A$,
					\item there is a consecutive pair of nodes $x,xa\in\inner(t)$ with $t(x)=(v,p)$ and $t(xa)=(w,q)$ defines a \textbf{kept transition} $(p,a,q)$ of $\A$,
				\end{enumerate}
				\item every consecutive pair of nodes $x,xa\in\inner(t)$ with
				\begin{enumerate}[label=(\arabic*)]
					\item $t(x)=(v,p)$ and $t(xa)=(w,q)$ where $(p,a,q)$ is not a kept transition defines an \textbf{insertion of a transition} $(p,a,q)$ into $\A$,				
					\item $t(x)=(v,p)$ and $t(xa)=(w,\#)$ defines an \textbf{insertion of a transition} $(p,a,xa)$ into $\A$,
					\item $t(x)=(v,\#)$ and $t(xa)=(w,q)$ defines an \textbf{insertion of a transition} $(x,a,q)$ into $\A$, and
					\item $t(x)=(v,\#)$ and $t(xa)=(w,\#)$ defines an \textbf{insertion of a transition} $(x,a,xa)$ into $\A$.
				\end{enumerate}
		\end{itemize}
	The final states of $\A$ are changed with respect to the edit-operations on states, that is:
		\begin{itemize}
			\item every kept state $q\in Q^\A$ with the respective tree node $x\in\inner(t)$ with $t(x)=(v,q)$ is 
			\begin{enumerate}[label=(\arabic*)]	
				\item an \textbf{unmarked final} state  $q$ of $\A$, if $q\in Q^{\A}_F$ and the node label $v=\eqclass{w}{L}$ such that $w\not\in L$, 
				\item a \textbf{kept final} state  $q$ of $\A$, if $q\in Q^{\A}_F$ and the node label $v=\eqclass{w}{L}$ such that $w\in L$,
				\item a \textbf{marked final} state $q$ of $\A$, if $q\not\in Q^{\A}_F$ and the node label $v=\eqclass{w}{L}$ such that $w\in L$, and
			
			\end{enumerate}
			\item every inserted state $x\in\inner(t)$ is an \textbf{inserted final} state $x$ into $\A$, if $t(x)=(v,\#)$ and $v=\eqclass{w}{L}$ such that $w\in L$.
		\end{itemize} 
		Finally, the initial states of $\A$ are changed with respect to the label of the root of $t$, that is:
		
		\begin{itemize}
			\item a state $q\in Q^{\A}$ is called 		 
			\begin{enumerate}[label=(\arabic*)]	
					\item an \textbf{unmarked initial} state of $\A$, if $q\in Q^{\A}_I$ and $t(\varepsilon)\neq(\varepsilon,q)$,
					\item a \textbf{kept initial} state of $\A$, if $q\in Q^{\A}_I$ and $t(\varepsilon)=(\varepsilon,q)$,
					\item a \textbf{marked initial} state of $\A$, if  $t(\varepsilon)=(\varepsilon,q)$, and
			\end{enumerate}
			\item if $t(\varepsilon)=(\varepsilon,\#)$, then $\varepsilon$ is the \textbf{inserted initial} state of $\A$.
		\end{itemize}
\end{definition}
Without going into technical details, observe that the edit-operations defined above are essentially a ``natural language'' description of the construction of the corrected automaton given in Definition~\ref{def:corrections}. However, there are two major differences between that perspective and Definition~\ref{def:corrections}. First of all, the source automaton is always determinized, as the correction trees representing \trimDFA{s}. And secondly, the resulting automata from the given description of edits is incomplete for the reason that one tree unfolding $t$, and therefore a correction tree, represents several possible target automata only differs in how transitions are inserted for the leaves of $t$ labeled with $\bot_{\eqclass{u}{L}}$. We resolve this ambiguity, as already addressed in Proposition~\ref{prop:dfa_reg_tree_lang} by picking one of the possibly represented target automata and marking all added transitions as \textbf{inserted}. 

As an example, consider the correction tree $t^{\corr}_{\B{\EX{2}}}$ shown in Figure~\ref{fix:example_correction_tree}. 
According to Definition~\ref{def:edit_ops}, this tree induces the following edit-operations on the source automaton $\A$ (left in Figure~\ref{fix:example_correction_tree}): The state $p_2$ is \textbf{deleted}, while the states $p_0$ and $p_1$ are kept, and a new final state $q_1$ is \textbf{inserted}. 
	Note that deleting the state $p_2$ also implicitly removes it from the set of final states. 
	Furthermore, $p_1$ becomes an \textbf{unmarked final} state, $p_0$ is an \textbf{unmarked initial} state, and $p_1$ is a \textbf{marked initial} state. 
	Regarding the transitions $(p_0,a,p_1)$, $(p_1,a,p_2)$, and $(p_2,a,p_2)$ are \textbf{deleted}, whereas the transition $(p_1,1,p_0)$ is kept. 
	In addition, the transitions $(p_1,0,q_1)$, $(q_1,a,p_2)$, and $(p_2,a,q_1)$ are \textbf{inserted}.
For readability, edit-operations associated with the sink state are omitted.

Edit-operations are classically measured via the Levenshtein distance on strings \cite{Levenshtein_SPD66}. In contrast, we do not allow replacement operations (cf. Definition~\ref{def:edit_ops}), although replacing a transition $(p,a,q)$ by $(p,a,q')$ could naturally be viewed as a single operation rather than as the deletion of $(p,a,q)$ followed by the insertion of $(p,a,q')$. Throughout this work, however, we restrict ourselves to insertions and deletions only.

Hence, we define a distance function $\Dist : \ALL{\A}{L} \to \mathbb{N}$ on correction trees as the number of all \emph{non-kept} edit-operations, ignoring edits involving the sink state for simplicity. For the example above, this yields $\Dist\left(t^{\corr}_{\B{\EX{2}}}\right)=11$.

\paragraph{Filter}
After elaborating how the information stored in a correction tree can be encoded as feedback, we turn to the question of computing the "right" feedback. Before doing so, note that the following discussion serves as a key motivation for the admittedly quite technical formalization of the \AllCorrect problem as a regular tree language. This representation enables the systematic filtering and selection of feedback candidates, as described below. Further note that the following is a technical description of the filtering process, where the presented filters serve as illustrative best-practice examples of the kinds of feedback selection that the framework supports, rather than an empirical analysis of their effect on student outcomes.

The method of filtering aligns with established techniques in fields such as the disambiguation of context-free grammars \cite{klint1994using}, XML query discovery \cite{niehren2013query}, data completion in tabular formats using existing cell values \cite{wang2017synthesis}, and program synthesis \cite{krogmeier2020decidable}, where a candidate space, represented as a regular tree language, is pruned using positive and negative examples or specific structural constraints.

Formally, a \textit{filter} is a subset $F \subseteq T^\mathcal{A}_L$ of the set of all corrections, typically defined by a structural or semantic constraint on correction trees. We say a filter $F$ is \textit{regular} if $F$ is a regular tree language, in which case $T^\mathcal{A}_L \cap F$ is again a regular tree language by closure under intersection. In the following, we define several didactically motivated filters and show that they are regular.

Let $\A$ be a source automaton, $\mathcal{T}^L_A$ the set of all corrections and $k \in \mathbb{N}$, then 
\begin{itemize}
	\item $\ED{k} := \{t\in \mathcal{T}^L_A \mid \Dist(t) \leq k\}$ denotes a filter limiting corrections to a maximal number of edit-operations,
	\item 	$\MINDFA := \{t\in \mathcal{T}^L_A \mid t = t_{\A^{\min}}, \textit{ where } \A^{\min} \textit{ is the minimal } \DFA \text{ of } L\}$ 
	denotes a filter selecting only corrections onto the minimal \DFA, and
	\item $\NATURAL := \left\{t\in \mathcal{T}^L_A \mid \text{for all } v \in \inner(t) \text{ with } t(v)=(\eqclass{w}{L},q) \text{ we have } q = \# \text{ or } f_{\lelex}(q) \in \eqclass{w}{L} \right\}$ 	
	
	denotes a filter that restricts corrections to those that do not use arbitrary state mappings, but only mappings consistent with the canonical structure induced by $f_{\lelex}$, thereby preserving the student's interpretation of the existing states.
\end{itemize}
There are some further natural candidates taken from well known results of the field of tree-edit-distances not further explored here, as for instance restricting the corrections in Definition~\ref{def:corrections} to those satisfying the properties of tree-edit-distance mappings \cite{tai1979treetotree}. 
However, we continue by sketching the regularity of the filters given above.

It is well known that every nondeterministic top-down tree automaton can be transformed into an equivalent bottom-up tree automaton without increasing the number of states \cite{TATA}. 
Hence, the filter $\ED{k}$ can be recognized by a bottom-up tree automaton equivalent to $\AT^{\corr}$, which counts the number of non-kept edit-operations along the tree and accepts exactly those trees for which this number does not exceed $k$.
The construction of $\AT^{\corr}$ already enforces the embedding of the prefix tree of the minimal \DFA\ (see Proposition~\ref{prop:dfa_reg_tree_lang}). 
Therefore, $\MINDFA$ can be recognized by additionally requiring that the automaton rejects all trees extending beyond this prefix tree.
Finally, $\NATURAL$ is regular since the condition on state mappings only depends on the locally checkable condition $f_{\lelex}(q) \in \eqclass{w}{L}$ of the tree, which can therefore be verified directly by the automaton.

Since all defined filters are regular tree languages and hence closed under intersection, arbitrary combinations of such filters can be used to realize different didactic objectives. 
The size of the resulting set of correction trees depends strongly on the chosen filter. For $\MINDFA$ it is finite with
$\sum_{j=0}^{\min(m,n)} \binom{n}{j} \frac{m!}{(m-j)!}$ elements, where $m$ is the number of states of the minimal \DFA\ and $n = |Q^{\A}|$; for $\ED{k}$ it is finite for every $k \in \mathbb{N}$; and for $\NATURAL$ it is infinite, since arbitrarily large target automata satisfying the mapping constraint may exist. 
However, even after combining multiple filters, the resulting language typically still contains more than one correction tree. In our current implementation\footnote{\url{https://github.com/maurice-herwig/alcep_dfa}}, a single correction is selected randomly from the filtered language.
From a didactic perspective, this is a natural direction for future work.

\paragraph{Practice}
After defining the filters and discussing their properties, we analyze their effects on a real test set.
Specifically, we consider the following task:
\begin{center}
	\begin{tabular}{l}
		Let $L_{\EX{2}} = \{xa^n \in  \{0, 1, a\}^*\mid x=0, \text{ if } n \text{ is even and } x = 1 \text{ otherwise}\}$ \\
		Construct a finite automaton $\A$ such that $L(\A) = L_{\EX{2}}$. 
	\end{tabular}
\end{center}
Note that this is the same language used in the example of Section~\ref{sec:corrections},
where the minimal \DFA\ for $L_{\EX{2}}$ was shown in Figure~\ref{fix:example_correction_tree} (Automaton $\B_{\EX{2}}$).
The task is a prototypical exercise in learning finite automata. It was assigned as a homework problem in the summer terms of $2020$ and $2021$ in a second-semester computer science bachelor's course on formal languages. Only $103$ out of $151$ submissions were correct, underscoring the importance of feedback. Although the exercise did not explicitly require a \DFA, $137$ students submitted one. Remarkably, $9$ of the $14$ non-\DFA\ submissions were incorrect.
The most common error, made by eight students, yielded an automaton for $L_{\EX{2}} \cup \{\varepsilon\}$. As this mistake can be corrected easily, we instead consider the second most common error (made by five students), shown in Figure~\ref{fig:example_dfas:2}. For readability, the sink state is omitted. Similarly, in the correction trees presented below, nodes associated to the sink state are indicated only by dotted lines and are excluded from the computation of $\Dist$.
\begin{figure}[t]
	\centering
	\begin{subfigure}[b]{0.45\textwidth}
		\centering
		\begin{tikzpicture}[node distance=16mm, every state/.style={minimum size=4mm, inner sep=2pt}]
			\node[state,initial,initial text=] (q0) at (0,0) {$p_0$};
			\node[state] (q1) at (1,0.4) {$p_1$};
			\node[state] (q2) at (1,-0.4) {$p_2$};
			\node[state] (q3) at (2.25,0.4) {$p_3$};
			\node[state,accepting] (q4) at (2.25,-0.4) {$p_4$};
			\node[state,accepting] (q5) at (3.5,0.4) {$p_5$};
			
			\path[->] 
			(q0) edge [above] node {$0$} (q1)
			(q0) edge [below] node {$1$} (q2)
			(q2) edge [bend left, above] node {$a$} (q4)
			(q4) edge [below] node {$a$} (q2)
			(q1) edge [above] node {$a$} (q3)
			(q3) edge [above] node {$a$} (q5)
			(q5) edge [below, bend left] node {$a$} (q3)
			;
		\end{tikzpicture}
		\caption{$\A^1_{\Sub}$}
		\label{fig:example_dfas:2}
	\end{subfigure}
	\hfill
	\begin{subfigure}[b]{0.45\textwidth}
		\centering
		\begin{tikzpicture}[node distance=16mm, every state/.style={minimum size=4mm, inner sep=2pt}]
			\node[state,initial,initial text= ] (q0) at (0,0) {$p_0$};
			\node[state,accepting] (q1) at (1.25,0) {$p_1$};
			
			\path[->] 
			(q0) edge [bend left, above] node {$a$} (q1)
			(q1) edge [bend left, below] node {$a$} (q0)
			;
		\end{tikzpicture}
		\caption{$\A^2_{\Sub}$}
		\label{fig:example_dfas:3}
	\end{subfigure}
	\caption{Two student submissions for the task of constructing a finite automaton recognizing $L_{\EX{2}}$.}
	\label{fig:example_dfas}
\end{figure}
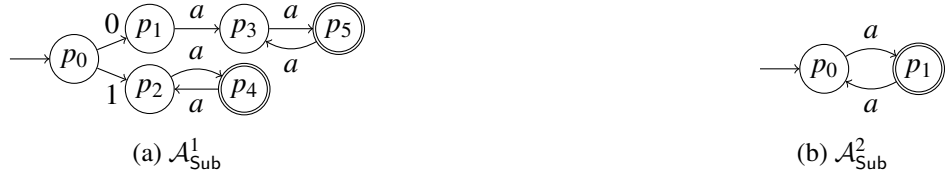
Figure~\ref{fig:filterexample1} shows the corrections obtained by applying different filters to this student submission, using the fixed alphabet order $a < 0 < 1$. Figure~\ref{fig:filterexample1:1} shows the only correction tree $t_1$ contained in $\ALL{\A^1_{\Sub}}{L_{\EX{2}}} \cap \ED{1} \cap \NATURAL$. In this correction, only state \textit{$p_1$ is marked as a final state}.

Figure~\ref{fig:filterexample1:2} illustrates one possible correction tree that transform $\A^1_{\Sub}$ into the minimal \DFA\ $\A$ (see Figure~\ref{fix:example_correction_tree}), whereas Figure~\ref{fig:filterexample1:3} shows the correction tree $t_3$ that transforms $\A^1_{\Sub}$ into $\A$ using the minimum number of edits.
The correction tree $t_2$ induces the following correction: \textit{delete states $p_1$, $p_3$, $p_4$, and $p_5$, together with the six transitions incident to these states; insert the final state $q_1$; and insert the transitions $(p_0, 0, q_1)$, $(q_1, a, p_2)$, and $(p_2, a, q_1)$}.
In contrast, the correction tree $t_3$ induces the following correction: \textit{delete states $p_1, p_3$ and $p_5$, together with the four transitions incident to these states; and insert the transition $(p_0, 0, p_4)$}. Consequently, we obtain $\Dist(t_2)=14$ and $\Dist(t_3)=8$.
\begin{figure}[t]
	\centering
	\begin{subfigure}[b]{0.35\textwidth}
		\resizebox{\linewidth}{!}{
			\begin{tikzpicture}[
				level distance=8mm,
				level 1/.style={sibling distance=40mm},
				level 2/.style={sibling distance=10mm},
				level 3/.style={sibling distance=10mm},
				every node/.style={
					draw,
					rounded rectangle,
					align=center,
					inner sep=2pt
				},
				empty/.style={
					draw=none,
					inner sep=0pt,
					minimum size=1pt
				},
				dotted child/.style={
					edge from parent/.style={draw,dotted}
				},
				edge from parent/.style={draw}
				]
				
				\node (e) {$\eqclass{\varepsilon}{}, p_0$}
				child {
					node (p1) {$\eqclass{0}{}, p_1$}
					child[dotted child] { node[empty] {} }
					child[dotted child] { node[empty] {} }
					child {
						node (p3l) {$\eqclass{1}{}, p_3$}
						child[dotted child] { node[empty] {} }
						child[dotted child] { node[empty] {} }
						child { node (b0) {$\bot_{\eqclass{0}{}}, \#$} }
					}
				}
				child {
					node (p2) {$\eqclass{1}{}, p_2$}
					child[dotted child] { node[empty] {} }
					child[dotted child] { node[empty] {} }
					child {
						node (p3r) {$\eqclass{0}{}, p_3$}
						child[dotted child] { node[empty] {} }
						child[dotted child] { node[empty] {} }
						child { node (b1) {$\bot_{\eqclass{1}{}},\#$} }
					}
				}
				child[missing];
				
				\draw[dotted]
				(e.south east)
				-- ++(4mm,-4mm);
			\end{tikzpicture}
		}
		\caption{$t_1 \in\ALL{\A^1_{\Sub}}{L_{\EX{2}}} \cap \ED{1} \cap \NATURAL$}
		\label{fig:filterexample1:1}
	\end{subfigure}
	\hfill
	\begin{subfigure}[b]{0.3\textwidth}
		\resizebox{\linewidth}{!}{
			\begin{tikzpicture}[
				level distance=8mm,
				level 1/.style={sibling distance=30mm},
				level 2/.style={sibling distance=10mm},
				every node/.style={
					draw,
					rounded rectangle,
					align=center,
					inner sep=2pt
				},
				empty/.style={
					draw=none,
					inner sep=0pt,
					minimum size=1pt
				},
				dotted child/.style={
					edge from parent/.style={draw,dotted}
				},
				edge from parent/.style={draw}
				]
				
				\node (e) {$\eqclass{\varepsilon}{}, p_0$}
				child {
					node (p1) {$\eqclass{0}{}, \#$}
					child[dotted child] { node[empty] {} }
					child[dotted child] { node[empty] {} }
					child { node (b1) {$\bot_{\eqclass{1}{}},\#$} }
				}
				child {
					node (p2) {$\eqclass{1}{}, p_2$}
					child[dotted child] { node[empty] {} }
					child[dotted child] { node[empty] {} }
					child { node (b1) {$\bot_{\eqclass{0}{}},\#$} }
				}
				child[missing];
				
				\draw[dotted]
				(e.south east)
				-- ++(4mm,-4mm);
				
			\end{tikzpicture}
		}
		\caption{$t_2 \in \ALL{\A^1_{\Sub}}{L_{\EX{2}}} \cap \MINDFA$}
		\label{fig:filterexample1:2}
	\end{subfigure}
	\hfill
	\begin{subfigure}[b]{0.3\textwidth}
		\resizebox{\linewidth}{!}{
			\begin{tikzpicture}[
				level distance=8mm,
				level 1/.style={sibling distance=30mm},
				level 2/.style={sibling distance=10mm},
				every node/.style={
					draw,
					rounded rectangle,
					align=center,
					inner sep=2pt
				},
				empty/.style={
					draw=none,
					inner sep=0pt,
					minimum size=1pt
				},
				dotted child/.style={
					edge from parent/.style={draw,dotted}
				},
				edge from parent/.style={draw}
				]
				
				\node (e) {$\eqclass{\varepsilon}{}, p_0$}
				child {
					node (p1) {$\eqclass{0}{}, p_4$}
					child[dotted child] { node[empty] {} }
					child[dotted child] { node[empty] {} }
					child { node (b1) {$\bot_{\eqclass{1}{}},\#$} }
				}
				child {
					node (p2) {$\eqclass{1}{}, p_2$}
					child[dotted child] { node[empty] {} }
					child[dotted child] { node[empty] {} }
					child { node (b1) {$\bot_{\eqclass{0}{}},\#$} }
				}
				child[missing];
				
				\draw[dotted]
				(e.south east)
				-- ++(4mm,-4mm);
				
			\end{tikzpicture}
		}
		\caption{$t_3 \in\ALL{\A^1_{\Sub}}{L_{\EX{2}}} \cap \ED{8} \cap \MINDFA$}
		\label{fig:filterexample1:3}
	\end{subfigure}
	\caption{Correction trees for $\A_{\Sub}^{1}$ for different filters.}
	\label{fig:filterexample1}
\end{figure}
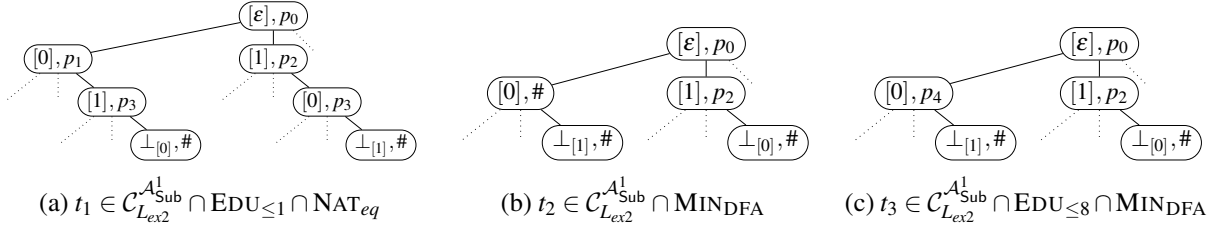

This example demonstrates that different filters may yield substantially different corrections and that, from a pedagogical perspective, correcting a student submission towards a fixed target solution, such as the minimal \DFA, is not always the most appropriate choice. Instead, $t_1$ preserves the student's intended construction and corrects only the underlying mistake.

However, combining the $\NATURAL$ filter with the minimum-edit criterion does not always yield the most appropriate correction. To illustrate this, we consider a second student submission, shown in Figure~\ref{fig:example_dfas:3}. Applying the $\NATURAL$ filter and selecting a correction tree with minimum edit distance results in the correction tree $t_4$, shown in Figure~\ref{fig:filterexample2:2}. The correction induced by $t_4$ retains only the initial state, deleting all remaining states and transitions and inserting the entire minimal \DFA.In contrast, the correction tree $t_5$ preserves the student's apparent intention by retaining the $aa$-loop.

These two examples demonstrate that no single filter consistently yields the didactically most suitable correction. Rather, the choice of an appropriate filter strongly depends on students' submission.
\begin{figure}[t]
	\centering
	\begin{subfigure}[b]{0.3\textwidth}
		\resizebox{\linewidth}{!}{
			\begin{tikzpicture}[
				level distance=7mm,
				level 1/.style={sibling distance=30mm},
				level 2/.style={sibling distance=10mm},
				every node/.style={
					draw,
					rounded rectangle,
					align=center,
					inner sep=2pt
				},
				empty/.style={
					draw=none,
					inner sep=0pt,
					minimum size=1pt
				},
				dotted child/.style={
					edge from parent/.style={draw,dotted}
				},
				edge from parent/.style={draw}
				]
				
				\node (e) {$\eqclass{\varepsilon}{}, p_0$}
				child {
					node (p1) {$\eqclass{0}{}, \#$}
					child[dotted child] { node[empty] {} }
					child[dotted child] { node[empty] {} }
					child { node (b1) {$\bot_{\eqclass{1}{}},\#$} }
				}
				child {
					node (p2) {$\eqclass{1}{}, \#$}
					child[dotted child] { node[empty] {} }
					child[dotted child] { node[empty] {} }
					child { node (b1) {$\bot_{\eqclass{0}{}},\#$} }
				}
				child[missing];
				
				\draw[dotted]
				(e.south east)
				-- ++(4mm,-4mm);
				
			\end{tikzpicture}
		}
		\caption{$t_5 \in\ALL{\A^2_{\Sub}}{L_{\EX{2}}} \cap \NATURAL$}
		\label{fig:filterexample2:2}
	\end{subfigure}
	\hfil
		\begin{subfigure}[b]{0.3\textwidth}
		\resizebox{\linewidth}{!}{
			\begin{tikzpicture}[
				level distance=7mm,
				level 1/.style={sibling distance=30mm},
				level 2/.style={sibling distance=10mm},
				every node/.style={
					draw,
					rounded rectangle,
					align=center,
					inner sep=2pt
				},
				empty/.style={
					draw=none,
					inner sep=0pt,
					minimum size=1pt
				},
				dotted child/.style={
					edge from parent/.style={draw,dotted}
				},
				edge from parent/.style={draw}
				]
				
				\node (e) {$\eqclass{\varepsilon}{}, \#$}
				child {
					node (p1) {$\eqclass{0}{}, p_1$}
					child[dotted child] { node[empty] {} }
					child[dotted child] { node[empty] {} }
					child { node (b1) {$\bot_{\eqclass{1}{}},\#$} }
				}
				child {
					node (p2) {$\eqclass{1}{}, p_0$}
					child[dotted child] { node[empty] {} }
					child[dotted child] { node[empty] {} }
					child { node (b1) {$\bot_{\eqclass{0}{}},\#$} }
				}
				child[missing];
				
				\draw[dotted]
				(e.south east)
				-- ++(4mm,-4mm);
				
			\end{tikzpicture}
		}
		\caption{$t_4 \in\ALL{\A^2_{\Sub}}{L_{\EX{2}}} \cap \MINDFA$}
		\label{fig:filterexample2:1}
	\end{subfigure}
	\caption{Correction trees for $\A_{\Sub}^{2}$ for different filters.}
	\label{fig:filterexample2}
\end{figure}
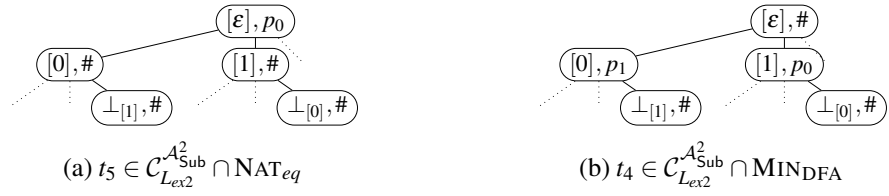

    \section{Conclusion}\label{sec:conclusion}
    Based on an encoding of all \trimDFA{}s recognizing a regular target language
$L$ as a regular tree language, we have shown that the set of all corrections that transform a given \DFA\ into a \DFA\ for
$L$ can also be characterized as a regular tree language. As this representation is primarily motivated by the use of corrections as a basis for didactic feedback in exercises on finite automata, it enables effective filtering of all conceptually valid ways in which an incorrect \DFA\ can be repaired, rather than relying on a single sample solution.
Furthermore, we have shown by example that the latter offers the possibility to model a student's solution intent as a regular tree language and thus, using closure properties, to compute personalized admissible feedback on a solution attempt.
However, as already stated in Section~\ref{sec:dfas_and_trees}, the results summarized above are currently limited to deterministic automata. Before turning to future research directions and applications, we briefly discuss the reasons for this limitation and possible ways to incorporate \NFA{}s into the given framework.

\paragraph{NFAs}

As mentioned before, the framework presented here for solving the \AllCorrect\ problem can be interpreted as a basis for extensions to other formalisms, for instance, from a didactically relevant perspective, logical theories or other types of automata such as \NFA{}s. Essentially, it relies on a finite canonical representation of all target instances of the given formalism. For \DFA{}s this is provided by the Nerode equivalence and the minimal \DFA\ derived from it. There are two obvious reasons why there is no straightforward extension of this framework to nondeterministic finite automata. First, \NFA{}s lack the property of having a unique canonical representative. Second, as already pointed out in Section~\ref{sec:dfas_and_trees}, the \DFA-equivalence problem can be reduced to the word problem of the regular tree language encoding of all \DFA{}s for a given language. From a complexity standpoint, this reduction lies well within the known polynomial complexity bounds for \DFA-equivalence. However, a direct translation of all \NFA{}s accepting a language $L$ into a regular tree language would enable a reduction of \NFA-equivalence, which is known to be a PSPACE-complete problem, to the word problem of regular tree languages, which is in P. 
The latter already hints at an obvious way of incorporating \NFA{}s into this framework, namely by simply choosing an exponential encoding by determinizing them. From a didactic perspective, however, this would require a more involved notion of edit-operations in order to correct a given source \NFA\ towards the intended \NFA\ rather than the corresponding \DFA.

Moreover, two further possible extensions to \NFA{}s directly address the need for a canonical representative in the above framework. For instance, residual automata (\RFA) were introduced for the same reason by restricting the definition of a \NFA\ $\A$ such that for every state $q$ there is a word $u$ with $L_q=u^{-1}L(\A)$. Hence, such automata immediately enable a similar tree encoding as the one introduced in Section~\ref{sec:dfas_and_trees} and furthermore, for every regular language $L$ there is a unique minimal residual automata \cite{DBLP:conf/stacs/DenisLT01}. However, besides the reason that computing the minimal residual automata is PSPACE-complete, from a didactical point of view, structurally correcting students solution towards \RFA\ seems not be the most ``natural'' and therefore suited feedback.

Another possible direction is to use the universal automaton (\UA) as a canonical structure. 
The \UA\ contains, as subautomata, many structurally distinct \NFA{s} recognizing the same language and, in particular, all minimal \NFA{s} \cite{lombardy2008universal}. Corrections could then be derived analogously to Definition~\ref{def:corrections} by mapping states of a student automaton into suitable subautomata of the \UA. Consequently, this approach restricts corrections to target automata whose number of states does not exceed the size of the \UA, which from a didactical perspective may still be sufficient, since the \UA\ already contains all minimal \NFA{s} recognizing the target language. However, the \UA\ may already be exponentially larger than the corresponding minimal \DFA\ for the same language \cite{lombardy2008universal}. Moreover, unlike in the deterministic setting, candidate subautomata additionally need to be checked for language equivalence, which for \NFA{s} is PSPACE-complete.

\paragraph{Future Work}
From a theoretical perspective and besides the above mentioned extensions to \NFA\ future work mainly concerns the extension of the above framework to related problems arising in educational settings, such as logical theories. 

From a didactical perspective, there are several research directions and applications of the set of all corrections presented here. Most importantly, many repair-based feedback systems for programming assignments \cite{price2016generating,rivers2017data} depend on prior student submissions or collections of reference solutions. The latter is known as the so-called \emph{cold-start problem}. Then the set of all corrections can be seen as an approach to overcome the cold-start problem.


Finally, the practical evaluation of the quality of the resulting feedback needs to be empirically investigated to assess how different classes of filtered corrections are perceived by students and how they influence understanding and learning outcomes, in comparison to boolean feedback, counterexamples, and conceptual hints \cite{d2015can}.

    %
    %
    %

    \bibliographystyle{eptcs}
    \bibliography{literatur}

\begin{thebibliography}{10}
\providecommand{\bibitemdeclare}[2]{}
\providecommand{\surnamestart}{}
\providecommand{\surnameend}{}
\providecommand{\urlprefix}{Available at }
\providecommand{\url}[1]{\texttt{#1}}
\providecommand{\href}[2]{\texttt{#2}}
\providecommand{\urlalt}[2]{\href{#1}{#2}}
\providecommand{\doi}[1]{doi:\urlalt{https://doi.org/#1}{#1}}
\providecommand{\eprint}[1]{arXiv:\urlalt{https://arxiv.org/abs/#1}{#1}}
\providecommand{\bibinfo}[2]{#2}

\bibitemdeclare{article}{aho1972minimum}
\bibitem{aho1972minimum}
\bibinfo{author}{Alfred \surnamestart Aho\surnameend} \&
  \bibinfo{author}{Thomas \surnamestart Peterson\surnameend}
  (\bibinfo{year}{1972}): \emph{\bibinfo{title}{{A minimum distance error
  correcting parser for context-free languages}}}.
\newblock {\slshape \bibinfo{journal}{SIAM Journal on Computing}}
  \bibinfo{volume}{1}(\bibinfo{number}{4}), pp. \bibinfo{pages}{305--312},
  \doi{10.1137/0201022}.

\bibitemdeclare{inproceedings}{bezakova2022effective}
\bibitem{bezakova2022effective}
\bibinfo{author}{Ivona \surnamestart Bez{\'a}kov{\'a}\surnameend},
  \bibinfo{author}{Kimberly \surnamestart Fluet\surnameend},
  \bibinfo{author}{Edith \surnamestart Hemaspaandra\surnameend},
  \bibinfo{author}{Hannah \surnamestart Miller\surnameend} \&
  \bibinfo{author}{David~E \surnamestart Narv{\'a}ez\surnameend}
  (\bibinfo{year}{2022}): \emph{\bibinfo{title}{{Effective succinct feedback
  for intro CS theory: A JFLAP extension}}}.
\newblock In: {\slshape \bibinfo{booktitle}{Proc.\ of the 53rd ACM Technical
  Symposium on Computer Science Education-Volume 1}}, \bibinfo{publisher}{ACM},
  pp. \bibinfo{pages}{976--982}, \doi{10.1145/3478431.3499416}.

\bibitemdeclare{inproceedings}{bruse2024computing}
\bibitem{bruse2024computing}
\bibinfo{author}{Florian \surnamestart Bruse\surnameend} \&
  \bibinfo{author}{Martin \surnamestart Lange\surnameend}
  (\bibinfo{year}{2024}): \emph{\bibinfo{title}{{Computing All Minimal Ways to
  Reach a Context-Free Language}}}.
\newblock In: {\slshape \bibinfo{booktitle}{Proc.\ 18th Int.\ Conf.\ on
  Reachability Problems, {RP'18}}}, {\slshape \bibinfo{series}{LNCS}}
  \bibinfo{volume}{15050}, \bibinfo{publisher}{Springer}, pp.
  \bibinfo{pages}{38--53}, \doi{10.1007/978-3-031-72621-7_4}.

\bibitemdeclare{book}{TATA}
\bibitem{TATA}
\bibinfo{author}{Hubert \surnamestart Comon\surnameend}, \bibinfo{author}{Max
  \surnamestart Dauchet\surnameend}, \bibinfo{author}{R{\'e}mi \surnamestart
  Gilleron\surnameend}, \bibinfo{author}{Florent \surnamestart
  Jacquemard\surnameend}, \bibinfo{author}{Denis \surnamestart
  Lugiez\surnameend}, \bibinfo{author}{Christof \surnamestart
  L{\"o}ding\surnameend}, \bibinfo{author}{Sophie \surnamestart
  Tison\surnameend} \& \bibinfo{author}{Marc \surnamestart Tommasi\surnameend}
  (\bibinfo{year}{2008}): \emph{\bibinfo{title}{{Tree Automata Techniques and
  Applications}}}.
\newblock \urlprefix\url{https://inria.hal.science/hal-03367725}.

\bibitemdeclare{inproceedings}{d2020automata}
\bibitem{d2020automata}
\bibinfo{author}{Loris \surnamestart D'Antoni\surnameend},
  \bibinfo{author}{Martin \surnamestart Helfrich\surnameend},
  \bibinfo{author}{Jan \surnamestart Kretinsky\surnameend},
  \bibinfo{author}{Emanuel \surnamestart Ramneantu\surnameend} \&
  \bibinfo{author}{Maximilian \surnamestart Weininger\surnameend}
  (\bibinfo{year}{2020}): \emph{\bibinfo{title}{{Automata tutor v3}}}.
\newblock In: {\slshape \bibinfo{booktitle}{Computer Aided Verification: 32nd
  International Conference, {CAV} 2020, Proc.\, Part {II}}},
  \bibinfo{publisher}{Springer}, pp. \bibinfo{pages}{3--14},
  \doi{10.1007/978-3-030-53291-8_1}.

\bibitemdeclare{article}{d2015can}
\bibitem{d2015can}
\bibinfo{author}{Loris \surnamestart D'Antoni\surnameend},
  \bibinfo{author}{Dileep \surnamestart Kini\surnameend},
  \bibinfo{author}{Rajeev \surnamestart Alur\surnameend},
  \bibinfo{author}{Sumit \surnamestart Gulwani\surnameend},
  \bibinfo{author}{Mahesh \surnamestart Viswanathan\surnameend} \&
  \bibinfo{author}{Bj{\"o}rn \surnamestart Hartmann\surnameend}
  (\bibinfo{year}{2015}): \emph{\bibinfo{title}{{How can automatic feedback
  help students construct automata?}}}
\newblock {\slshape \bibinfo{journal}{ACM Transactions on Computer-Human
  Interaction (TOCHI)}} \bibinfo{volume}{22}(\bibinfo{number}{2}), pp.
  \bibinfo{pages}{1--24}, \doi{10.1145/2723163}.

\bibitemdeclare{article}{d2015automata}
\bibitem{d2015automata}
\bibinfo{author}{Loris \surnamestart D'Antoni\surnameend},
  \bibinfo{author}{Matthew \surnamestart Weavery\surnameend},
  \bibinfo{author}{Alexander \surnamestart Weinert\surnameend} \&
  \bibinfo{author}{Rajeev \surnamestart Alur\surnameend}
  (\bibinfo{year}{2015}): \emph{\bibinfo{title}{{Automata Tutor and what we
  learned from building an online teaching tool}}}.
\newblock {\slshape \bibinfo{journal}{Bulletin of EATCS}}
  \bibinfo{volume}{3}(\bibinfo{number}{117}).

\bibitemdeclare{inproceedings}{DBLP:conf/stacs/DenisLT01}
\bibitem{DBLP:conf/stacs/DenisLT01}
\bibinfo{author}{Fran{\c{c}}ois \surnamestart Denis\surnameend},
  \bibinfo{author}{Aur{\'{e}}lien \surnamestart Lemay\surnameend} \&
  \bibinfo{author}{Alain \surnamestart Terlutte\surnameend}
  (\bibinfo{year}{2001}): \emph{\bibinfo{title}{{Residual Finite State
  Automata}}}.
\newblock In: {\slshape \bibinfo{booktitle}{{STACS} 2001, 18th Annual Symposium
  on Theoretical Aspects of Computer Science}}, \bibinfo{series}{Lecture Notes
  in Computer Science}, \bibinfo{publisher}{Springer}, pp.
  \bibinfo{pages}{144--157}, \doi{10.1007/3-540-44693-1\_13}.

\bibitemdeclare{inproceedings}{herwig2025earley}
\bibitem{herwig2025earley}
\bibinfo{author}{Maurice \surnamestart Herwig\surnameend},
  \bibinfo{author}{Norbert \surnamestart Hundeshagen\surnameend} \&
  \bibinfo{author}{Martin \surnamestart Lange\surnameend}
  (\bibinfo{year}{2025}): \emph{\bibinfo{title}{{An Earley-Based Universal
  Error-Correcting Parser}}}.
\newblock In: {\slshape \bibinfo{booktitle}{International Conference on
  Implementation and Application of Automata}}, \bibinfo{publisher}{Springer},
  pp. \bibinfo{pages}{208--222}, \doi{10.1007/978-3-032-02602-6_15}.

\bibitemdeclare{book}{hopcroft2001introduction}
\bibitem{hopcroft2001introduction}
\bibinfo{author}{John~E \surnamestart Hopcroft\surnameend},
  \bibinfo{author}{Rajeev \surnamestart Motwani\surnameend} \&
  \bibinfo{author}{Jeffrey~D \surnamestart Ullman\surnameend}
  (\bibinfo{year}{2001}): \emph{\bibinfo{title}{{Introduction to automata
  theory, languages, and computation}}}.
\newblock \bibinfo{publisher}{Addison-Wesley}.

\bibitemdeclare{inproceedings}{klint1994using}
\bibitem{klint1994using}
\bibinfo{author}{Paul \surnamestart Klint\surnameend} \& \bibinfo{author}{Eelco
  \surnamestart Visser\surnameend} (\bibinfo{year}{1994}):
  \emph{\bibinfo{title}{{Using filters for the disambiguation of context-free
  grammars}}}.
\newblock In: {\slshape \bibinfo{booktitle}{Proc. ASMICS Workshop on Parsing
  Theory}}, \bibinfo{address}{Milan, Italy}, pp. \bibinfo{pages}{1--20}.

\bibitemdeclare{inproceedings}{krogmeier2020decidable}
\bibitem{krogmeier2020decidable}
\bibinfo{author}{Paul \surnamestart Krogmeier\surnameend},
  \bibinfo{author}{Umang \surnamestart Mathur\surnameend},
  \bibinfo{author}{Adithya \surnamestart Murali\surnameend},
  \bibinfo{author}{P~\surnamestart Madhusudan\surnameend} \&
  \bibinfo{author}{Mahesh \surnamestart Viswanathan\surnameend}
  (\bibinfo{year}{2020}): \emph{\bibinfo{title}{{Decidable synthesis of
  programs with uninterpreted functions}}}.
\newblock In: {\slshape \bibinfo{booktitle}{International Conference on
  Computer Aided Verification}}, \bibinfo{publisher}{Springer}, pp.
  \bibinfo{pages}{634--657}, \doi{10.1007/978-3-030-53291-8_32}.

\bibitemdeclare{article}{Levenshtein_SPD66}
\bibitem{Levenshtein_SPD66}
\bibinfo{author}{V.~I. \surnamestart Levenshtein\surnameend}
  (\bibinfo{year}{1966}): \emph{\bibinfo{title}{{Binary codes capable of
  correcting deletions, insertions and reversals}}}.
\newblock {\slshape \bibinfo{journal}{Soviet Physics Doklady}}
  \bibinfo{volume}{10}(\bibinfo{number}{8}), pp. \bibinfo{pages}{707--710}.
\newblock \bibinfo{note}{Doklady Akademii Nauk SSSR, V163 No4 845-848 1965}.

\bibitemdeclare{incollection}{lombardy2008universal}
\bibitem{lombardy2008universal}
\bibinfo{author}{Sylvain \surnamestart Lombardy\surnameend} \&
  \bibinfo{author}{Jacques \surnamestart Sakarovitch\surnameend}
  (\bibinfo{year}{2008}): \emph{\bibinfo{title}{{The universal automaton}}}.
\newblock In: {\slshape \bibinfo{booktitle}{Logic and Automata}},
  \bibinfo{publisher}{Amsterdam University Press}, pp.
  \bibinfo{pages}{457--504}.

\bibitemdeclare{article}{niehren2013query}
\bibitem{niehren2013query}
\bibinfo{author}{Joachim \surnamestart Niehren\surnameend},
  \bibinfo{author}{J{\'e}r{\^o}me \surnamestart Champav{\`e}re\surnameend},
  \bibinfo{author}{Aur{\'e}lien \surnamestart Lemay\surnameend} \&
  \bibinfo{author}{R{\'e}mi \surnamestart Gilleron\surnameend}
  (\bibinfo{year}{2013}): \emph{\bibinfo{title}{{Query induction with
  schema-guided pruning strategies}}}.
\newblock {\slshape \bibinfo{journal}{The Journal of Machine Learning
  Research}} \bibinfo{volume}{14}(\bibinfo{number}{1}), pp.
  \bibinfo{pages}{927--964}, \doi{10.5555/2567709.2502611}.

\bibitemdeclare{article}{pillay2010learning}
\bibitem{pillay2010learning}
\bibinfo{author}{Nelishia \surnamestart Pillay\surnameend}
  (\bibinfo{year}{2010}): \emph{\bibinfo{title}{{Learning difficulties
  experienced by students in a course on formal languages and automata
  theory}}}.
\newblock {\slshape \bibinfo{journal}{ACM SIGCSE Bulletin}}
  \bibinfo{volume}{41}(\bibinfo{number}{4}), pp. \bibinfo{pages}{48--52},
  \doi{10.1145/1709424.1709444}.

\bibitemdeclare{inproceedings}{price2016generating}
\bibitem{price2016generating}
\bibinfo{author}{Thomas~W \surnamestart Price\surnameend},
  \bibinfo{author}{Yihuan \surnamestart Dong\surnameend} \&
  \bibinfo{author}{Tiffany \surnamestart Barnes\surnameend}
  (\bibinfo{year}{2016}): \emph{\bibinfo{title}{{Generating data-driven hints
  for open-ended programming}}}.
\newblock In: {\slshape \bibinfo{booktitle}{Proc.\ of the 9th International
  Conference on Educational Data Mining}}, \bibinfo{publisher}{International
  Educational Data Mining Society}, pp. \bibinfo{pages}{191--198}.

\bibitemdeclare{article}{rivers2017data}
\bibitem{rivers2017data}
\bibinfo{author}{Kelly \surnamestart Rivers\surnameend} \&
  \bibinfo{author}{Kenneth~R \surnamestart Koedinger\surnameend}
  (\bibinfo{year}{2017}): \emph{\bibinfo{title}{{Data-driven hint generation in
  vast solution spaces: a self-improving python programming tutor}}}.
\newblock {\slshape \bibinfo{journal}{International Journal of Artificial
  Intelligence in Education}} \bibinfo{volume}{27}, pp.
  \bibinfo{pages}{37--64}, \doi{10.1007/s40593-015-0070-z}.

\bibitemdeclare{inproceedings}{robson2024fsm}
\bibitem{robson2024fsm}
\bibinfo{author}{Eliot~Wong \surnamestart Robson\surnameend},
  \bibinfo{author}{Sam \surnamestart Ruggerio\surnameend} \&
  \bibinfo{author}{Jeff \surnamestart Erickson\surnameend}
  (\bibinfo{year}{2024}): \emph{\bibinfo{title}{{{FSM} Builder: {A} Tool for
  Writing Autograded Finite Automata Questions}}}.
\newblock In: {\slshape \bibinfo{booktitle}{Proc.\ of the 2024 on Innovation
  and Technology in Computer Science Education V. 1, {ITiCSE} 2024}},
  \bibinfo{publisher}{{ACM}}, pp. \bibinfo{pages}{269--275},
  \doi{10.1145/3649217.3653599}.

\bibitemdeclare{inproceedings}{shekhar2014enhancing}
\bibitem{shekhar2014enhancing}
\bibinfo{author}{Vinay~S \surnamestart Shekhar\surnameend},
  \bibinfo{author}{Anant \surnamestart Agarwalla\surnameend},
  \bibinfo{author}{Akshay \surnamestart Agarwal\surnameend},
  \bibinfo{author}{B~\surnamestart Nitish\surnameend} \& \bibinfo{author}{Viraj
  \surnamestart Kumar\surnameend} (\bibinfo{year}{2014}):
  \emph{\bibinfo{title}{{Enhancing JFLAP with automata construction problems
  and automated feedback}}}.
\newblock In: {\slshape \bibinfo{booktitle}{2014 seventh international
  conference on contemporary computing (ic3)}}, \bibinfo{publisher}{IEEE}, pp.
  \bibinfo{pages}{19--23}, \doi{10.1109/IC3.2014.6897141}.

\bibitemdeclare{inproceedings}{singh2013automated}
\bibitem{singh2013automated}
\bibinfo{author}{Rishabh \surnamestart Singh\surnameend},
  \bibinfo{author}{Sumit \surnamestart Gulwani\surnameend} \&
  \bibinfo{author}{Armando \surnamestart Solar-Lezama\surnameend}
  (\bibinfo{year}{2013}): \emph{\bibinfo{title}{{Automated feedback generation
  for introductory programming assignments}}}.
\newblock In: {\slshape \bibinfo{booktitle}{Proc.\ of the 34th ACM SIGPLAN
  conference on Programming language design and implementation}},
  \bibinfo{publisher}{ACM}, pp. \bibinfo{pages}{15--26},
  \doi{10.1145/2491956.2462195}.

\bibitemdeclare{article}{tai1979treetotree}
\bibitem{tai1979treetotree}
\bibinfo{author}{Kuo-Chung \surnamestart Tai\surnameend}
  (\bibinfo{year}{1979}): \emph{\bibinfo{title}{{The Tree-to-Tree Correction
  Problem}}}.
\newblock {\slshape \bibinfo{journal}{J. ACM}}
  \bibinfo{volume}{26}(\bibinfo{number}{3}), pp. \bibinfo{pages}{422--433},
  \doi{10.1145/322139.322143}.

\bibitemdeclare{article}{wang2017synthesis}
\bibitem{wang2017synthesis}
\bibinfo{author}{Xinyu \surnamestart Wang\surnameend}, \bibinfo{author}{Isil
  \surnamestart Dillig\surnameend} \& \bibinfo{author}{Rishabh \surnamestart
  Singh\surnameend} (\bibinfo{year}{2017}): \emph{\bibinfo{title}{{Synthesis of
  data completion scripts using finite tree automata}}}.
\newblock {\slshape \bibinfo{journal}{Proc.\ of the ACM on Programming
  Languages}} \bibinfo{volume}{1}(\bibinfo{number}{OOPSLA}), pp.
  \bibinfo{pages}{1--26}, \doi{10.1145/3133886}.

\end{thebibliography}

\end{document}